\documentclass[twocolumn,notitlepage,nofootinbib,tightenlines]{revtex4-1}

\usepackage{hyperref}
\usepackage{amsmath}
\usepackage{amsthm}
\usepackage{amsfonts}

\hypersetup{breaklinks=true}

\makeatletter

\renewcommand*{\p@subsection}{}

\renewcommand*{\p@subsubsection}{}
\makeatother

\newtheorem{thm}{Theorem}[section]

\theoremstyle{definition}

\theoremstyle{remark}

\numberwithin{equation}{section}

\begin{document}

\title{Maxwell's equations are universal\\ for locally conserved quantities}

\author{Lucas Burns}
\address{Chapman University}
\email{luburns@chapman.edu}
\date{\small May 30, 2019}

\begin{abstract}
A fundamental result of classical electromagnetism is that Maxwell's equations imply that electric charge  is  locally  conserved.  Here  we  show the converse: Local  charge  conservation  implies  the local existence of fields satisfying Maxwell's equations. This holds true for any conserved quantity satisfying a continuity equation. It is obtained by means of a strong form of the Poincar\'e lemma presented here that states: Divergence-free multivector fields locally possess curl-free antiderivatives on flat manifolds. The above converse is an application of this lemma in the case of divergence-free vector fields in spacetime. We also provide conditions under which the result generalizes to curved manifolds.
\end{abstract}

\maketitle

\section{Introduction}

Historically, Maxwell's equations of electromagnetism have taken various forms. They were first proposed by James Clerk Maxwell in 1865 as a set of twenty equations \cite{maxwell}. The very last of these, Maxwell called the equation of continuity, in analogy to the equation of mass continuity in hydrodynamics. In his original treatise, this was written as
\begin{equation}\label{eq:continuity-1864}
    \frac{de}{dt} + \frac{df}{dx} + \frac{dg}{dy} + \frac{dh}{dz} = 0.
\end{equation}
Today's standard treatments of electromagnetism tend to be expressed in vector calculus. For instance, in John D. Jackson's \emph{Classical Electrodynamics} (1998) or David J. Griffith's \emph{Introduction to Electrodynamics} (2012), you will find the familiar four equation expression for Maxwell's equations:
\begin{align}\label{eq:maxwell-vector}
    \begin{split}
        \vec \nabla \cdot \vec E &= \frac{1}{\epsilon_0} \rho\\
        \vec \nabla \times \vec B &= \mu_0 \vec J + \frac{1}{\mu_0 \epsilon_0} \frac{\partial \vec E}{\partial t}
    \end{split}
    \begin{split}
        \vec \nabla \cdot \vec B &= 0\\
        \vec \nabla \times \vec E &= - \frac{\partial \vec B}{\partial t},
    \end{split}
\end{align}
which, with some fiddling, can be shown to \emph{imply} the continuity equation
\begin{equation}\label{eq:continuity-vector}
    \frac{\partial \rho}{\partial t} + \vec \nabla \cdot \vec J = 0.
\end{equation}
The fact that the continuity equation follows from Maxwell's equations (in particular, from the inhomogenous equations in the left column) is a fundamental result of electrodynamics, and it is crucial to the theory, because it means that electric charge is locally conserved. This can be seen by re-expressing Equation~\ref{eq:continuity-vector} as
\begin{equation}
  \frac{\partial Q}{\partial t} = \frac{\partial}{\partial t} \int_\mathcal{V} \rho dV = - \int_{\partial\mathcal{V}} \vec J \cdot d \vec a,
\end{equation}
which states: The total electric charge $Q$ in a region of space $\mathcal{V}$ can only change if it flows through the boundary $\partial \mathcal{V}$ of that region.

The purpose of this paper is to demonstrate that Equation~\ref{eq:continuity-vector} also implies the local existence of fields satisfying Equations~\ref{eq:maxwell-vector}. The main mathematical result utilized in order to demonstrate this, which we call the strong Poincar\'e lemma, is a marriage of the Poincar\'e lemma of de Rham cohomology and the integral formula of geometric calculus. Though their connection appears simple in hindsight, the implications seem underexplored. In the language of geometric calculus, it states: Divergence-free multivector fields locally possess curl-free antiderivatives (or dually, curl-free multivector fields locally possess divergence-free antiderivatives), under certain conditions. Unlike the usual Poincar\'e lemma, the strong lemma is dependent on a metric. To my knowledge, this result has been shown once before in Ref. \cite{brackx}. We present a simplified derivation, present conditions under which this result remains valid on arbitrary manifolds, and demonstrate its useful application for conservation laws.

This lemma applied to the case of divergence-free vector fields in spacetime provides as an immediate consequence that Maxwell's equations are universal for locally conserved quantities. To my knowledge, the argument that Maxwell's equations can be obtained from the continuity equation has been made by two others \cite{heras, macdonald}. This paper reinforces these results by demonstrating that it holds without assumption of particular boundary conditions, clarifies the extent to which it is true in topologically non-trivial spacetimes, and provides conditions under which the result also holds in curved spacetimes. Moreover, the strong lemma yields what might be called generalized Maxwell equations from generalized conservation laws. For this reason, we emphasize that Maxwell's equations are not unique to electromagnetism and may be of use in the analysis of other locally conserved quantities. As a first example, we are investigating the use of this lemma to offer mathematical justification for an analogy to electromagnetism utilized in recent work on acoustic waves \cite{bliokh1, bliokh2, shi, burns}.

We begin with an overview of the Poincar\'e lemma, its expression in geometric calculus, the integral formula, a proof of the strong lemma, and follow with its application to locally conserved quantities. We conclude with a brief discussion of implications.

\section{The Poincar\'e lemma}

John Baez and Javier P. Muniain in their text \emph{Gauge Fields, Knots, and Gravity} (1994) write Maxwell's equations in the language of differential forms as
\begin{align}\label{eq:maxwell-forms}
    \begin{split}
        \star d \star F = J
    \end{split}
    \begin{split}
        d F = 0,
    \end{split}
\end{align}
the first of which implies the continuity equation
\begin{equation}\label{eq:continuity-forms}
    d \star J = 0,
\end{equation}
where the electromagnetic field $F$ is now represented as a 2-form, the current density $J$ as a 1-form, $d$ is the exterior derivative, and $\star$ is the Hodge star on some spacetime manifold.

The fact that Equation~\ref{eq:continuity-forms} follows from Equations~\ref{eq:maxwell-forms} follows from the property of the exterior derivative $d^2 = 0$ or its adjoint $(\star d \star)^2 = 0$.  In particular, if we knew the conditions under which $(\star d \star) J = 0$ implied the existence of a $2$-form $F$ such that $J = (\star d \star) F$, then we would be halfway to showing that Maxwell's equations follow from the continuity equation---only missing the homogenous equations $d F = 0$.

This is precisely the subject of de Rham cohomology, which asks the question: When are closed differential forms exact --- where a $k$-form $\alpha$ is said to be closed if $d \alpha = 0$ and exact if there exists a $k-1$-form $\beta$ such that $\alpha = d \beta$?

When an $n$-dimensional manifold $\mathcal{M}$ is equipped with a metric, then a hodge star operator $\star$ can be defined that maps between $k$-forms and $n-k$-forms and yields the dual question: When are co-closed differential forms co-exact --- where a $k$-form $\alpha$ is co-closed if $(\star d \star) \alpha = 0$ and co-exact if there exists a $k+1$-form $\beta$ such that $\alpha = (\star d \star) \beta$? Notice that the exterior derivative $d$ raises the grade of forms, and its adjoint $\star d \star$ lowers grade.

Interestingly, the answer to these two questions depends strictly on the topology of $\mathcal{M}$. We will only utilize one small result of this theory, but the importance of working out the full details of this theory in geometric calculus should be noted. For a text that develops de Rham theory in detail, see Ref. \cite{warner}. The result we will use is the Poincar\'e lemma, which states the following.
\begin{thm}[Poincar\'e Lemma for Forms]\label{thm:forms}
    If $\mathcal{M}$ is a smooth, differentiable manifold, then closed differential forms are locally exact on $\mathcal{M}$.
\end{thm}
\noindent Locally here means in contractible neighborhoods of points in open sets of $\mathcal{M}$, which exist for sufficiently small neighborhoods. See Ref. \cite{spivak} for a proof\footnote{Ref. \cite{spivak} proves this for differentiable manifolds embedded in Euclidean space. The theorem presented here follows for smooth, differentiable manifolds due to the Whitney embedding theorem.}. This implies the following dual lemma.
\begin{thm}[Dual Poincar\'e Lemma for Forms]\label{thm:dual-forms}
    If $\mathcal{M}$ is a smooth, differentiable manifold equipped with a metric, then co-closed differential forms are locally co-exact on $\mathcal{M}$.
\end{thm}
\begin{proof}
  Consider a co-closed $n-k$-form $\rho$ in some contractible region of $\mathcal{M}$, satisfying $\star d \star \rho = 0$. Then $\rho = \star \alpha$ is dual to some $k$-form $\alpha$, which is closed: $d \alpha = d \star \rho = 0$. By Theorem~\ref{thm:forms}, there exists some $k-1$-form $\beta$ satisfying $\alpha = d \beta$. This implies that $\rho = \star \alpha = \star d \beta = \star d (\star \sigma$), where $\beta = \star \sigma$, for some $\sigma$ dual to $\beta$. Thus $\rho = (\star d \star) \sigma$ is co-exact.
\end{proof}
What this tells us is that Equation~\ref{eq:continuity-forms} implies the existence of a 2-form satisfying the first equation in Equations~\ref{eq:maxwell-forms}. Indeed, Fredriech Hehl and Yuri Obukhov use this fact to obtain the inhomogenous Maxwell equations $(\star d \star) F = J$ from the continuity equation $(\star d \star) J = 0$ in their premetric approach to electrodynamics \cite{hehl}.

Below we will examine precisely the conditions under which there exists a 2-form that satisfies both Equations~\ref{eq:maxwell-forms}. We begin by laying the groundwork for this examination. First we will present a brief overview of the machinery of geometric calculus and translate the above lemmas for use on vector manifolds where we can utilize the integral formula, then we will move on to proving the strong lemma and consider its application to local conservation laws.

\section{Geometric calculus on vector manifolds}

Maxwell's equations take the following form in geometric calculus
\begin{equation}\label{eq:maxwell-fields}
    D F = J
\end{equation}
Here $F$ is a bivector field, $J$ is a vector field, and $D$ is the covariant derivative on a 4-dimensional vector manifold. This encompasses both of Equations~\ref{eq:maxwell-forms} because the geometric product, which tells us $aM = a \cdot M + a \wedge M$ for any vector $a$ and multivector $M$, unifies the exterior derivative $d$ and its adjoint $\star d \star$ into a single covariant derivative operator $D$. We will discuss the relationship between operators in more detail below.

A vector manifold $\mathcal{M}$ is a manifold the points of which are vectors \cite{hestenes, doran}. One way to construct such a manifold is to embed it as a surface in a higher dimensional flat space. There is an intrinsic approach as well \cite{hestenes-shape}, however we will take extrinsic perspective in this paper. Vector manifolds are sufficiently general to describe smooth Riemannian manifolds, which are the primary manifolds of interest in physics \cite{doran}.

There are two derivative operators on vector manifolds that are important for our purposes. First, the vector derivative operator $\partial$, which can be regarded as the projection of the ambient, flat space derivative operator onto the tangent space of $\mathcal{M}$ \cite{macdonald, hestenes, doran}. Secondly, the covariant derivative $D$
\begin{equation}\label{eq:vector-covariant}
  D A = \partial A - S(A),
\end{equation}
defined for any multivector field $A$ that lives in the tangent space of $\mathcal{M}$, where $S(A)$ is called the shape operator and encodes important information about the curvature of $\mathcal{M}$ \cite{hestenes, doran}. The exterior derivative and its adjoint are related to $D$.

Notice that the continuity equation in geometric calculus
\begin{equation}\label{eq:continuity-fields}
    D \cdot J = 0,
\end{equation}
follows from Equation~\ref{eq:maxwell-fields} due to the fact that the divergence of a divergence $D \cdot (D \cdot M) = 0$ vanishes for all multivector fields $M$ on $\mathcal{M}$. This is also true for the curl of a curl, $D \wedge (D \wedge M) = 0$. It turns out that these are equivalent to $(\star d \star)^2 = 0$ and $d^2 = 0$, respectively.

In particular, any $r$-form $\alpha_r$ can be written in terms of an $r$-vector:
\begin{align}\label{eq:form-field}
    \alpha_r = A_r \cdot dX_r^\dagger,
\end{align}
where $dX_r = I_r |dX_r|$ is a \emph{directed} measure with grade $r$.

Importantly, the exterior derivative of $\alpha_r$ is equivalent to the curl:
\begin{align}\label{eq:exterior-curl}
    d \alpha_r = (D \wedge A_r) \cdot dX_{r+1}^\dagger,
\end{align}
and the adjoint is equivalent to the divergence (up to a sign):
\begin{align}\label{eq:adjoint-divergence}
    (-1)^{n(r+1)+1}(\star d \star) \alpha_r = (D \cdot A_r) \cdot dX_{r-1}^\dagger.
\end{align}

See Chapter~6.4 of Ref. \cite{doran} or Section~6-5 of Ref. \cite{hestenes} for more details. These equivalences allow us to work freely with $r$-vectors and the covariant derivative $D$ in place of $r$-forms, the exterior derivative $d$, and its adjoint.

Using this correspondence, the Poincar\'e lemma and its dual can be expressed in the language of geometric calculus as follows.
\begin{thm}[Poincar\'e Lemma for Fields]\label{thm:fields}
    If $\mathcal{M}$ is a smooth, differentiable vector manifold, then curl-free fields are locally the curl of a field.
\end{thm}
\begin{proof}
  Let $F$ be a multivector field on manifold $\mathcal{M}$ of dimension $n$ and $F_r = \langle F \rangle_r$ be the grade $r$ part of $F$. If $D \wedge F = 0$, then $D \wedge F_{r} = \langle D \wedge F \rangle_{r+1} = 0$ for each $r$. Utilizing Theorem~\ref{thm:forms} and Equation~\ref{eq:exterior-curl}, we have that there exists some field $A_{r-1}$ of grade $r-1$ such that $F_r = D \wedge A_{r-1}$. This implies that $F = D \wedge A$, for $A = \sum_{r=0}^{n} A_{r}$.
\end{proof}
\noindent
The dual is obtained by an analogous proof.
\begin{thm}[Dual Poincar\'e Lemma for Fields]\label{thm:dual-fields}
    If $\mathcal{M}$ is a smooth, differentiable vector manifold, then divergence-free fields are locally the divergence of a field.
\end{thm}
\section{Strong Poincar\'e lemma}

The only remaining result we need before turning to the strong lemma is the integral formula of geometric calculus. For a detailed exposition of the integral formula, see Section~7-3 of Ref. \cite{hestenes}. We simply present the result.
\begin{thm}[The Integral Formula]\label{thm:integral-formula}
    Let $F$ be a field, integrable on a simple\footnote{Without self-intersections.}, $n$-dimensional vector manifold $\mathcal{M}$. Then it possesses antiderivatives $A$ with respect to the vector derivative $\partial A = F$, determined by $F$ up to boundary conditions, given by
    \begin{align}\label{eq:integral-formula}
        (-1)^m I(x) A(x) = &- \int g(x, x') d^m x' F(x') \\
        &+ \oint g(x, x') d^{m-1} x' A(x'),\nonumber
    \end{align}
    where $I=I(x)$ is the unit pseudoscalar field over $\mathcal{M}$ and $g$ is a Green's function of $\partial$.
\end{thm}
All the usual subtleties are present in actual computation of this integral. For instance, on spacetime manifolds, null surfaces cannot be used as boundaries. See Section~6.3 of Ref. \cite{doran} for discussion of this issue. Antiderivatives differ at most by a monogenic field\footnote{Monogenic fields $\psi$ on $\mathcal{M}$ are fields that satisfy $\partial \psi = 0$ and are fully determined by boundary conditions due to Theorem~\ref{eq:integral-formula}. Notice that for two-dimensional monogenic functions, Equation~\ref{eq:integral-formula} reduces to Cauchy's integral formula of complex analysis. See Section~6.4 of Ref. \cite{doran} or Section~7-4 of Ref. \cite{hestenes} for details.} satisfying $\partial \psi = 0$. Once boundary conditions for $A$ are specified on $\mathcal{M}$, Equation~\ref{eq:integral-formula} yields a unique antiderivative. This gives us all we need to prove the strong lemma.

\begin{thm}[Strong Poincar\'e Lemma]\label{thm:strong}
    If $\mathcal{M}$ is a flat manifold, then curl-free fields have divergence-free antiderivatives locally.
\end{thm}
\begin{proof}
    Let $F$ be a curl-free field on $\mathcal{M}$, such that $D \wedge F = 0$. Then locally, there exists a field $A$ on $\mathcal{M}$ such that $F = D \wedge A$ by Theorem~\ref{thm:fields}. By Theorem~\ref{eq:integral-formula}, $A$ possesses a local antiderivative $B$ such that $A = \partial B$. Since $\mathcal{M}$ is flat, we have $A = D B$, which implies
    \begin{align}
        \begin{split}
            F &= D \wedge A \\
              &= D \wedge (D \cdot B + D \wedge B) \\
              &= D \wedge (D \cdot B) \\
              &= D (D \cdot B),
        \end{split}
    \end{align}
    where we have used the facts $D \wedge (D \wedge B) = 0$ and $D \cdot (D \cdot B) = 0$. Thus $D \cdot B$ is a divergence-free, covariant antiderivative of $F$.
\end{proof}
\noindent
The dual is obtained by an analogous proof.
\begin{thm}[Strong Dual Poincar\'e Lemma]\label{thm:dual-strong}
    If $\mathcal{M}$ is a flat manifold, then divergence-free fields have curl-free antiderivatives locally.
\end{thm}
These theorems are restricted to flat manifolds because Theorem~\ref{eq:integral-formula} only provides antiderivatives with respect to the vector derivative $\partial$, and $D = \partial$ only holds on flat manifolds. If there exists a construction analogous to Theorem~\ref{eq:integral-formula} for covariant antiderivatives, then Theorems \ref{thm:strong} and \ref{thm:dual-strong} immediately generalize to curved manifolds.

A starting point for determining whether Theorem~\ref{thm:strong} generalizes to curved manifolds may be to consider the fundamental theorem of geometric calculus \cite{hestenes, doran} in terms of the covariant derivative and shape operator, as given by Equation~\ref{eq:vector-covariant}.

\begin{align}
    \int \dot g dX \dot \partial \dot A = \int \dot g dX (\dot D + \dot S) \dot A = \oint g dS A,
\end{align}
where $S$ is the shape operator. If we take $g$ to be a Green's function for the covariant derivative $D$, instead of the vector derivative, as is done in derivation of Equation~\ref{eq:integral-formula} \cite{hestenes}, we arrive at the formula
\begin{align}\label{eq:covariant-integral-formula}
    (-1)^m I A = &- \int g dX D A + \oint g dS A\\
    &- (-1)^m \int dX (g S (A) - S(g) A). \nonumber
\end{align}

It appears that we're stuck here, because the right-hand side includes integrals of $A$ on $\mathcal{M}$, which is what we're trying to determine with knowledge of $A$ on $\partial{\mathcal{M}}$ and $DA$ on $\mathcal{M}$. However, Equation~\ref{eq:covariant-integral-formula} is worth consideration, because for any field $F$ on $\mathcal{M}$, $S(F)$ lives outside of the tangent space of $\mathcal{M}$, so it is possible that projecting Equation~\ref{eq:covariant-integral-formula} onto $\mathcal{M}$ eliminates some contributions from the integrals involving the shape operator. For now, we will leave the conditions under which one can compute covariant antiderivatives (with respect to $D$) as an open question. Determining these conditions would be of great interest to understanding the extent to which Theorem~\ref{thm:strong} extends to curved manifolds.

Curiously, Theorem~\ref{thm:strong} appears to be rare in existing literature. One instance in Euclidean space can be found in Ref. \cite{brackx}. We've shown that the strong lemma holds on arbitrary flat manifolds, which is not immediately obvious due to its dependence on the metric, and presented conditions for its generalization to curved manifolds.

\section{Locally conserved quantities}

We now consider application to locally conserved quantities.
\begin{thm}[Maxwell's Equation for Conserved Vector Fields]\label{thm:conserved-currents} Local conservation laws of the form
    \begin{equation}
        D \cdot J = 0,
    \end{equation}
    where $J$ is a vector field on flat spacetime, imply the existence of an antiderivative $F$ satisfying Maxwell's equations
    \begin{equation}
        D F = D \cdot F = J.
    \end{equation}
\end{thm}
\begin{proof}
    Apply Theorem~\ref{thm:dual-strong} to $J$.
\end{proof}

Note that this is precisely Equation~\ref{eq:maxwell-fields} when $J$ is an electric current density. This reinforces the result of Ref. \cite{heras} which argues that local charge conservation is sufficient to obtain Maxwell's equation and can serve as a fundamental principle in axiomatic approaches to electrodynamics. The proof given here eliminates the need to assume particular boundary conditions. If Theorem~\ref{thm:dual-strong} holds on curved manifolds, then Theorem~\ref{thm:conserved-currents} would also generalize to curved spacetimes.

Conserved quantities of this form include local charge conservation in electrodynamics, local mass conservation in continuum mechanics, and local probability conservation in quantum mechanics. Note that this theorem also applies directly to conserved tensors without modification. We rephrase for the sake of being explicit.

\begin{thm}[Maxwell's Equation for Conserved Tensors]\label{thm:conserved-tensors} Local conservation laws that can be expressed as
\begin{equation}
    D \cdot T^\nu = D_\mu T^{\mu \nu} = 0
\end{equation}
for tensors $T^{\mu\nu}$, imply the existence of a bivector-valued antiderivative $F^\nu = \frac{1}{2} \gamma_\mu \wedge \gamma^\rho F^{\mu \nu}_{\rho}$ for each $T^\nu$ such that
\begin{equation}
    D F^\nu = D \cdot F^\nu = T^\nu.
\end{equation}
\end{thm}

Conserved quantities of this form include the gravitational stress energy tensor, where $F^\nu$ plays the role of a gravitational superpotential. Whether there are obstructions to applying the lemma to the conserved currents of Yang-Mills theory is left as an open question and is of particular interest, given the close structural relationship between Yang-Mills equations and Maxwell's equations.

\section{Discussion}

We have presented two main results: a strong form of the Poincar\'e lemma (Theorem~\ref{thm:strong}) and its application to conserved currents (Theorem~\ref{thm:conserved-currents}). While these results are restricted to flat manifolds, we've presented conditions under which they generalize to curved manifolds --- namely, the existence of antiderivatives for fields with respect to the covariant derivative $D$.

Note that the usual Poincar\'e lemma (Theorem~\ref{thm:fields}) doesn't tell us that a conserved current $J$ is exclusively the divergence of a bivector, just that there exists such a bivector. In the same way, the strong Poincar\'e lemma presented here doesn't tell us that an antiderivative $F$ \emph{must} satisfy Maxwell's equations, just that there exists such an $F$. In this view, the result may not be particularly surprising to physicists for the following reason.

Consider the more common situation of applying the Poincar\'e lemma to the electromagnetic field $F$. We use the fact that $D \wedge F = 0$ to determine the existence of a potential $A$ satisfying $F = D \wedge A$. Here, we assume that we can choose the Lorenz gauge $D \cdot A = 0$ --- it is a gauge \emph{freedom} after all. Theorem~\ref{thm:strong} tells us that there is no obstruction to making this choice.

In a certain sense, $D \wedge F = 0$ is also a gauge choice, insofar as the current density $J$ is concerned. That is, adding a divergence-free bivector to $F$ does not change the physical content of $J$, so in this way it is like a gauge transformation with respect to $J$.

On the other hand, this lemma means that Maxwell's equation in electrodynamics can be understood as expression of local charge conservation. It still carries the freedom to admit magnetic sources, but Maxwell's equation with both electric and magnetic sources can always be decoupled into a pair of Maxwell equations---one for electric sources and one for magnetic sources---so long as they are independently conserved.

Of course, Theorem~\ref{thm:conserved-currents} tells us nothing of the \emph{dynamics} of electromagnetism, since it tells us nothing of the force exerted by the field on the current. However, it does tell us that the force law is what distinguishes different electromagnetic theories in trivial topologies and helps us to know that topological theories of electrodynamics are fundamentally distinct, only resembling the classical theory locally.

As such, this result is helpful in guiding investigations of extensions to electrodynamics. For instance, some equations may appear to be generalizations of Maxwell's equations, but are not. Consider $\nabla F + \nabla \chi = J$, where $\chi$ is a scalar field satisfying $\nabla^2 \chi = 0$, as seen in Ref. \cite{dvoe} with $J=0$. The strong lemma tells us that solutions to this equation have corresponding solutions to Maxwell's equations, so are not generalizations but rather are reparameterizations.

Moreover, this result applies to any conserved current and may offer insight into theories beyond electrodynamics. Corresponding to any locally conserved quantity---and thus, by Noether's theorem, to any continuous symmetry---is a field satisfying Maxwell's equation.

Moreover, there are many theories which utilize analogies to electromagnetism. Theorem~\ref{thm:dual-strong} may offer formal grounding for such analogies. Consider for instance the gravitoelectromagnetic approach to linearized gravity\cite{mashhoon} which has been used to visualize the dynamics of merging black hole binary systems in the Simulating Extreme Spacetimes (SXS) project\cite{owen}.

In viewing Maxwell's equations as an expression of a conserved quantity, it is less surprising that these analogies between electromagnetism and general relativity exist, and perhaps also less mysterious that Yang-Mills theories have structure so closely resembling that of electrodynamics. These analogies are not coincidental. They arise naturally from the structure of locally conserved quantities.

\section*{Acknowledgements}

I am grateful to Justin Dressel for constructive conversations and important corrections, as well as the organizers of AGACSE 2018 for a wonderful conference.

\Urlmuskip=0mu plus 1mu
\bibliographystyle{plain}
\bibliography{main}

\end{document}